\documentclass[a4paper]{article}

\usepackage{mathtools}
\usepackage{amsmath}
\usepackage[amsmath,thmmarks]{ntheorem}
\usepackage{subcaption}

\usepackage{varioref}

\usepackage{datetime}

\usepackage{mathtools}

\usepackage[h]{esvect}

\usepackage{array}

\usepackage{epsfig}

\usepackage{listings}
\usepackage[activate]{pdfcprot}

\usepackage{booktabs}

\renewcommand{\epsilon}{\ensuremath\varepsilon}

\renewcommand{\phi}{\ensuremath{\varphi}}

\newcommand{\Lagr}{\mathcal{L}}
\labelformat{equation}{(#1)}

\usepackage{algorithm}
\usepackage[noend]{algpseudocode}

\renewcommand{\paragraph}[1]{\noindent \textbf{#1}.}

\usepackage{amssymb}
\newenvironment{proof}{\textit{Proof}.}{\hfill$\square$}

\newtheorem{theorem}{Theorem}

\newtheorem{corollary}{Corollary}

\title{The Transactional Conflict Problem}
\author{Dan Alistarh \\ IST Austria
\and Syed Kamran Haider \\ UConn
\and Raphael K{\"u}bler \\ ETH Zurich
\and Giorgi Nadiradze \\ ETH Zurich
}

\date{}

\begin{document}

\maketitle


\begin{abstract}
  The \emph{transactional conflict} problem arises in transactional systems whenever two or more concurrent transactions clash on a data item. 
  While the standard solution to such conflicts is to immediately abort one of the transactions, some practical systems consider the alternative of \emph{delaying} conflict resolution for a short interval, which may allow one of the transactions to commit. The challenge in the transactional conflict problem is to choose the optimal length of this delay interval so as to minimize the overall running time penalty for the conflicting transactions. 

   In this paper, we propose a family of optimal online algorithms for the transactional conflict problem. 
   Specifically, we consider variants of this problem which arise in different implementations of transactional systems, namely ``requestor wins'' and ``requestor aborts'' implementations:  
   in the former, the recipient of a coherence request is aborted, whereas in the latter, it is the requestor which has to abort. 
   Both strategies are implemented by real systems. 
   We show that the \emph{requestor aborts} case can be reduced to a classic instance of the ski rental problem, while the \emph{requestor wins} case leads to a \emph{new} version of this classical problem, for which we derive optimal deterministic and randomized algorithms. 
   Moreover, we prove that, under a simplified adversarial model, our algorithms are constant-competitive with the offline optimum in terms of throughput.
   We validate our algorithmic results empirically through a hardware simulation of hardware transactional memory (HTM), showing that our algorithms can lead to non-trivial performance improvements for classic concurrent data structures.  
\end{abstract}


\newcommand{\package}{\emph}
\DeclarePairedDelimiter\floor{\lfloor}{\rfloor}

\section{Introduction}

Several interesting tools and techniques have been developed to address the challenge of making correct decisions about the future while holding limited information, also known as \emph{online decision-making}~\cite{kalai_vempala, albers_scheduling, self_organizing, albers_caching, kserver}. In this context, the \emph{ski rental problem}~\cite{karlin} is the question of choosing between continuing to pay a recurring cost, and paying a larger one-time cost which eliminates or reduces the recurring cost.  
The standard setting of the problem is as follows. A person goes skiing for an unknown number of days $D$. On each day he or she can decide whether or not to buy skis for a fixed price $B$. If the person decides not to buy the skis, he or she has to rent them for a fixed cost per day $i < B$, where $i$ is usually equal to 1. The challenge is to find an algorithm that minimizes the expected cost of the ski tour over the number of days $D$. 

In this paper, we phrase the problem of resolving conflicts efficiently in hardware transactional systems~\cite{herlihy_moss} as an online decision problem. 
We abstract this task as the following \emph{transactional conflict} problem. 
Consider a database or hardware transactional system, which allows concurrent executions of transactions. 
The following scenario often arises in practice: a transaction $T_1$ accesses a set of data items, and is assigned ownership of a subset of these data items (usually the items in its write set). 
Assume that $T_1$ is executing concurrently with another transaction $T_2$, whose dataset intersects with that of $T_1$. 
In this case, this conflict of ownership is detected by the transactional system at runtime, and the transactional implementation usually makes one of two choices to resolve the conflict. 
The first strategy, usually called \emph{requestor wins}~\cite{stanford_tm, req_wins}, will have the transaction $T_2$ abort the transaction $T_1$, while $T_2$ takes ownership of the contended data item. 
The second strategy, called \emph{requestor aborts}~\cite{stanford_tm}, has transaction $T_2$ abort, resolving the conflict in favor of $T_1$. 

To mitigate the high number of aborts caused by such conflicts, a number of hardware proposals, e.g.~\cite{req_wins, lease_release} allow the following strategy: instead of aborting one of the transactions immediately, we allow the transaction a \emph{grace period} $\Delta$, \emph{before} aborting it. 
From the algorithmic perspective, assume our objective is to maximize the likelyhood of a commit, while minimizing the expected value of the delay added to the running time of the two transactions. 
By introducing a grace period, we increase the running time of the delayed transaction $T_2$. 
Yet, since aborts are expensive, if $T_1$ commits during the grace period, we may gain in terms of the \emph{sum of runtime costs}, since we could significantly   \emph{decrease} the total running time of $T_1$ if we do not abort it. 
Clearly, given this cost model, the question of optimally setting the grace period $\Delta$ is an online decision problem. 

\paragraph{Contribution} Our results are as follows. 
\begin{itemize}
	\item We formalize the \emph{transactional conflict problem} for both the \emph{requestor wins} and the \emph{requestor aborts} conflict resolution strategies (Section~\ref{sec:def}). 
	\item We give optimal deterministic and randomized solutions for both variants (Section~\ref{sec:non_const}). Further, we  
	analyze the performance of both schemes in the setting where \emph{additional information} is provided, in the form of the mean of the underlying distribution of transaction lengths, and when \emph{chains} of transactions might conflict. 
	We provide closed-form optimal solutions in this case (Section~\ref{sec:constrained}).
	
	\item We show that, under an adversarial transaction scheduling model, our algorithms are \emph{globally} constant competitive in terms of the sum of  running times of the transactions, i.e.  inverse of throughput. (Section~\ref{sec:competitive}) Further, we discuss strategies to add probabilistic progress guarantees to our algorithms (Section~\ref{sec:progress}). 
	\item We validate our strategy on a hardware simulator implementation of HTM~\cite{herlihy_moss} and in synthetic tests, showing that it can yield  throughput improvements of up to $4\times$, and that it is competitive with a finely-tuned manual approach (Section~\ref{exp}).
\end{itemize} 

\paragraph{The Transactional Conflict Problem} In more detail, the transactional conflict problem presents the following trade-off. (Please see Figure~\ref{fig:conflict} for an illustration.) 
Assume transaction $T_1$ is interrupted by a conflicting transaction $T_2$, in a \emph{requestor wins} implementation. 
We have the choice of aborting $T_1$ immediately, or delaying the abort by a grace period $\Delta$, in the hope that the transaction will commit before $\Delta$ expires. 
If the transaction $T_1$ \emph{commits} after some additional time $x < \Delta$,  then we have added $x$ time steps to the running time of $T_2$; yet, since $T_1$ commits, we do not incur any additional time penalty for this transaction. 
If the transaction $T_1$ has not committed within the $\Delta$ grace period, then we will abort it. In this case, we wasted $2\Delta$ time units in terms of the total running time of the transactions: one $\Delta$ for which we have delayed $T_2$, and one $\Delta$ for which we have increased the total running time of $T_1$, without committing it. (Since $T_1$ will have to redo all its computation in case of abort, none of the executed steps is useful.) 
In addition, we assume that we always incur a (large) cost $B$ for having aborted the transaction. 

In sum, given the unknown remaining running time $x$ of the transaction $T_1$, the task is to compute the grace period $\Delta$ which optimizes the trade-off between the cost $x$ we pay in the case where $x < \Delta$,  and the cost $2\Delta + B$ which we pay in the case where $\Delta \leq x$. This problem clearly appears to be related to the ski rental problem, and it is tempting to think that its solution follows by simple reduction. 

However, this is not the case, due to the \emph{different structure of the cost function}. 
Perhaps surprisingly, we show that this difference significantly alters the optimal strategy, as well as the competitive ratio. 
For example, for two transactions and a \emph{requestor wins} implementation, we prove that the optimal strategy is  a \emph{uniform random} choice in the interval $[0, B)$, and its competitive ratio is $2$. 
By contrast, the optimal strategy for \emph{requestor aborts} implementations, which coincides with the classic ski rental problem, is the exponentially decaying probability distribution, with competitive ratio $e / (e - 1)$. 
This trade-off becomes more complex for conflict chains involving more than two transactions. We also consider this case in this paper. 

\paragraph{Extensions and Techniques} 
We consider and solve two natural extensions of the problem. 
The first is to consider larger conflict chains, when more than two distinct transactions may be involved in a conflict.  The second is the case where the length of each transaction is assumed to come from an arbitrary probability distribution, whose mean $\mu$ is \emph{known}. For instance, this corresponds to a profiler which records the empirical mean over all successful executions of a transaction, and uses this information when deciding the grace period length.  

%
%
%

On the technical side, we give a unified technique for optimally solving parametrized instances of this online decision problem; in particular,  we provide closed-form solutions for the general case with arbitrarily many conflicting transactions and mean constraints. 
Our main analytic tool is a non-trivial instance of the method of Lagrange multipliers tailored to this problem, building on an analysis by Khanafer et al.~\cite{constrained_skirental} for the ski rental problem. 

\paragraph{Order-Optimality} 
Starting from the competitiveness of local decisions, we show that, under a simplified adversarial conflict model for a system of $n$ threads executing concurrent transactions,  the algorithms we propose are globally competitive with an offline-optimal solution in terms of the sum of running times of transactions (intuitively, the inverse of throughput). 
Specifically, if the timing of conflicts between transactions is decided by an adversary, then we  prove that the running time penalty incurred by all transactions using our scheme is \emph{constant-competitive} with a perfect-information algorithm, which knows the remaining running time of each transaction at conflict time.

\paragraph{Experimental Results}
We validate our results empirically, using both a synthetic testbed, and an implementation of a requestor-wins hardware transactional memory (HTM) system on top of the MIT Graphite processor simulator~\cite{graphite}, executing data structure and transactional benchmarks. 
The synthetic experiments closely match our theoretical results. The HTM experiments suggest that adding delays can improve performance in HTM implementations under contention, and does not adversely impact performance in uncontended ones. 
Of note, we observe that our algorithms perform well even when compared with a finely-tuned approach, which uses knowledge about the application and implementation to manually fine-tune the delays.

\paragraph{Implications} 
The general problem of \emph{contention management} in the context of transactional memory has been considered before~\cite{ghp05}, and several efficient strategies are known, e.g.~\cite{ghp05, DSTM}. 
The main distinction from our setting is that contention managers (for instance in software TM) are usually assumed to have global knowledge about the set of running transactions, and possibly their abort history. By contrast, in our setting, decisions are entirely \emph{local}, as well as immediate and unchangeable.  
In this context, our algorithms implement a \emph{distributed, online} contention manager. 
We find it surprising that constant-competitive throughput bounds can be obtained in this restrictive setting. 

A second implication of our work concerns \emph{requestor wins} versus \emph{requestor aborts} conflict resolution. 
Most systems only implement one such strategy, but our analytic results show a trade-off between contention and the performance of these two paradigms. 
In particular, \emph{requestor aborts} is more efficient under low contention, whereas \emph{requestor wins} is more efficient when conflicts involve more than two transactions. This suggests that a hybrid strategy, which can alternate between the two, would perform best. 

\section{Related Work}

The ski rental problem was introduced by Karlin, Manasse, Rudolph, and Sleator in~\cite{karlin86}, where the authors give optimal deterministic $2$-competitive algorithms for the classic version of the ski rental problem. 
Follow-up work by Karlin et al.~\cite{karlin} showed that randomized algorithms can improve on this competitive ratio to $e / (e - 1)$, which is optimal~\cite{seiden}. 
Several variants of this problem have been studied, e.g.~\cite{albers_caching, karlin2001dynamic, dooly1998tcp, constrained_skirental}, with strong practical applications, e.g.~\cite{dooly1998tcp, constrained_skirental}. 
We refer the reader to the excellent survey by Albers~\cite{albers} for an overview of the area, and to Section~\ref{sec:background} for technical background on the problem. 

The work that is technically closest to ours is by Khanafer et al.~\cite{constrained_skirental}: they define a \emph{constrained} version of the ski rental problem, 
in which constraints such as the mean and variance of the underlying distribution from which the unknown parameter is chosen are added. 
The authors provide bounds on the improvement in terms of competitive ratio provided by these additional constraints. 
The basic technical ingredients we employ, notably the method of Lagrange multipliers, are similar to those of this work. 
However, we consider a significantly more complex instance of the problem, which requires a more involved analysis. 

The problem of designing efficient contention management in the context of transactional memory~\cite{herlihy_moss} has received significant research attention, especially since such systems are already present in hardware by major vendors~\cite{intel_rtm}. 
Several applied papers, e.g.~\cite{stanford_tm, req_wins_vs_aborts, req_wins, pathologies} discuss the trade-offs involved in implementing transactional  protocols in hardware, and the performance pathologies of such systems. 
Reference~\cite{req_wins} explicitly discusses the possibility of adding delays to hardware transactions, and implements and compares heuristics for tuning these delays. 
We believe we are the first to abstract the problem of setting the delays while minimizing running time impact, and to solve it optimally. 
As noted, the contention management problem has been abstracted by~\cite{ghp05} in the context of software TM (STM), and there is a considerable amount of work on this topic, e.g.~\cite{ghp05, DSTM, attiya, yoo2008}. 
However, to our understanding, all these systems assume a contention manager module with global knowledge about the set of running transactions, which is reasonable in the context of STM; 
by contrast, we analyze a setting where only local, immediate, and unchangeable  decisions are possible, which is the case for HTM.

\section{Model and Preliminaries}
\label{sec:model}
\subsection{Shared-Memory Transactions}

We consider an asynchronous shared-memory model in which $n$ processors communicate by performing atomic operations to shared memory. We assume that standard read-write atomic operations are available to the processors. 

We are interested in a setting where processors execute sequences of basic read-write operations \emph{atomically}, as \emph{transactions}. 
Specifically, a transaction is a set of read and write operations, which are guaranteed to be executed atomically if the transaction \emph{commits}. Otherwise, if the transaction \emph{aborts}, then none of the operations is applied to shared memory. 
Many hardware implementations of transactional memory use some version of the pattern shown in Algorithm 1 below.

\begin{algorithm}[h]
	\caption{Simplified Pseudocode for Hardware Transactional Memory implementation.}
	\label{fig:general-alg}
	\begin{algorithmic}[1]
		
		\Function{transaction}{code}
		\State // \emph{Use a MESI cache coherence protocol,  except each cache line has an additional bit. This additional bit is set if cache line is used by transaction. In this case, cache line is called transactional and it resides in the transactional cache. } 

		\State // \emph{Execution phase}
		\State \textbf{if} at any point transactional cache line is evicted, abort transaction.
		\State \textbf{if} transaction is aborted, invalidate all transactional cache lines.

		\For{each access in code} 
		\If{read}
			\State Try to acquire the transactional cache line in a shared or exclusive state.
			\State \textbf{Conflict} arises if some transaction has the required cache line in its transactional cache in a modified state.
		\EndIf
		\If{write} 
			\State Try to invalidate all existing copies of the required cache line, then change a state of the cache line to modified.
			 \State \textbf{Conflict} arises if some transaction has a copy of the line in its tranactional cache.	
		\EndIf
		\EndFor
		
		\State // \emph{Commit phase}

		\State \textbf{if} transaction is not aborted in the execution phase, then \textbf{commit}, by clearing additional bits in all transactional cache lines.
		
		\EndFunction 
	\end{algorithmic}
\end{algorithm}

Notice that, if several transactions with overlapping read and write sets may execute or attempt to commit at the same time, then they may conflict in the commit phase. In particular, consider the following example, involving transactional objects $A, B$ and $C$: 

A transaction $T_1$ performing $[ A \gets B + 1; C = 3 ]$ and a transaction $T_2$ performing $[A \gets B + 2]$ from an initial state of $A = B = 0$ can conflict in the commit phase as follows. Imagine $T_1$ holds variable $A$ in Exclusive state locally, and is in the process of acquiring variable $B$ in Exclusive state, so that it can commit. Then $T_2$ starts executing its Commit phase, and sends a coherence message to $T_1$, asking for $A$ in Exclusive state. At this point, $T_1$ has two choices: it can either relinquish exclusive access of $A$, and abort its transaction (see Figure \ref{fig:req_wins}), or delay the request by $A$ for some time, hoping that it will commit soon (see Figure \ref{fig:req_aborts}). 
In the following, we will look at algorithms for making this choice. Note that we will use both terms algorithm and strategy meaning the same thing from now on.

\begin{figure}
	\centering
	\begin{subfigure}[b]{0.45\textwidth}
		\includegraphics[width=\textwidth]{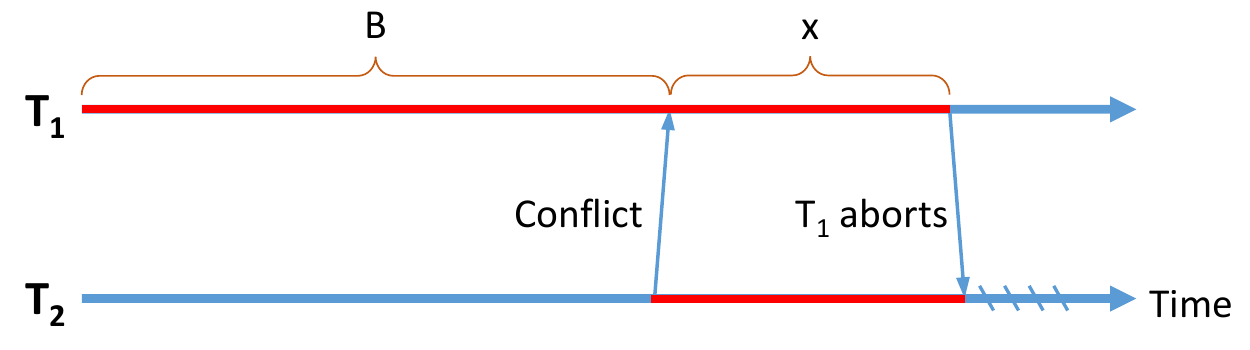}
		\caption{Requestor wins scenario. Here, $T_1$ has been running for $B$ time steps when a conflict with $T_2$ is detected. 
	$T_1$ has the possibility of delaying $T_2$'s message response for $x$ steps before aborting itself. The ``wasted'' time added to the execution of both transactions in case $T_1$ aborts is shown in red. }
	\label{fig:req_wins}
\end{subfigure} ~
	\begin{subfigure}[b]{0.45\textwidth}
	\epsfig{width=\textwidth, file=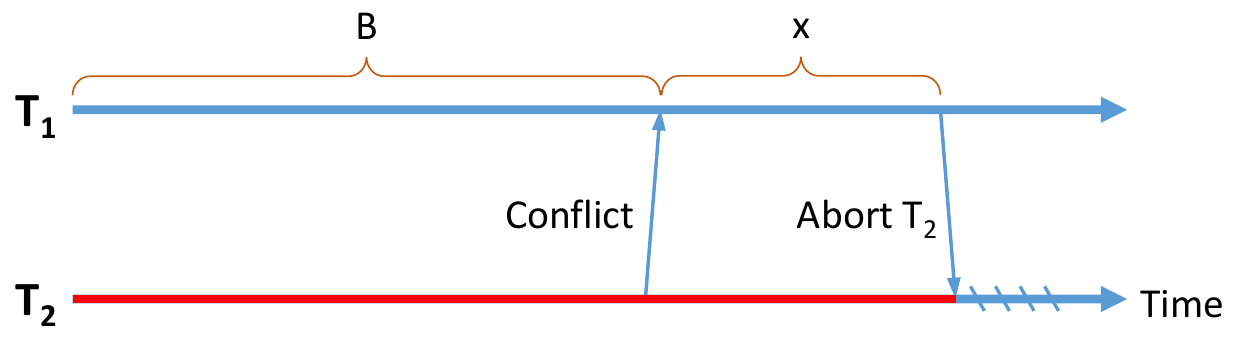}
	\caption{Requestor aborts scenario. Again, $T_1$ has been running for $B$ time steps when a conflict with $T_2$ is detected. 
	$T_1$ has the possibility of delaying $T_2$'s message response for $x$ steps before aborting $T_2$. The ``wasted'' time added to the execution of both transactions in case  $T_2$ aborts is shown in red. }
	\label{fig:req_aborts}
\end{subfigure}
	\caption{The transactional conflict problem for the two conflict resolution strategies.}
	\label{fig:conflict}
\end{figure}

\subsection{Conflict and Cost Models} 

\paragraph{Transaction Conflict Scheduling} As described, we have $n$ threads, each of which has a virtually infinite sequence of transactions to execute. 
Threads proceed to execute these transactions in parallel. 
At arbitrary times during the execution, an adversary can interrupt a pair of transactions,  and put them in conflict: one of them will be the \emph{requestor}, and the other one is the \emph{receiver}. 
The algorithm has the choice of resolving the conflict immediately, by aborting the one of the transactions, or to postpone the abort for a grace period $\Delta$. 
The cost model is specified for each conflict resolution strategy in Section~\ref{sec:def}.

Upon aborting a transaction, the thread will restart its execution immediately. 
Initially, the conflicted transaction will no longer be in conflict with any other transaction, although it may become conflicted if chosen by the adversary. 
Upon committing a transaction, the thread moves to the next transaction in its input. 
During its grace period, a receiver transaction may become conflicted as a requestor with another transaction, since it may need to access some new data item. 

\paragraph{Additional Assumptions} In the following, we will assume a simplified version of the conflict model, and analyze the global competitiveness of our strategies. In particular, we assume that (a) a transaction that is part of a conflict as a requestor cannot become part of a new conflict as a receiver,  (b) a transaction that is currently during its grace period cannot be conflicted again as a receiver by the adversary (but the transaction may be conflicted as a requestor), and (c) that conflicts cannot be cyclic. 
Assumptions (a) and (b) are in some sense necessary given our strong adversarial model: we wish to ensure that the adversary can only inflict the same set of conflicts on the offline optimal strategy as to the online decision algorithm. 
We note that assumption (c) is implemented by some real-world HTM implementations, which actively detect conflict cycles, and abort all transactions involved upon such events~\cite{WCH}.

\paragraph{Cost Model}
Given the above setup, our cost model is defined as follows. Fix an adversarial conflict strategy $S$, and a (possibly randomized) algorithm $\mathcal{A}$ for resolving conflicts. For each transaction $T$, define $\Gamma(T, \mathcal{A})$ be the expected length of the interval between $T$'s starting time (i.e., the first time it is invoked) and its eventual commit time, assuming conflicts induced by $S$.  
We analyze strategies for deciding the grace period, which minimize the \emph{expected sum of running times of transactions}. Formally, we wish to find an algorithm $\mathcal{A}$ such that, for any adversarial strategy $S$, 
$ \sum_{T} \Gamma(T, \mathcal{A}) \textnormal{ is minimized.} $

\subsection{Background on the Ski Rental Problem}
\label{sec:background}
 We recall the definition of the ski rental problem~\cite{karlin}: a person goes skiing for an unknown number of days $D$. On each day, he or she can decide whether or not to buy skis for a fixed price $B$. 
 If the person decides to not buy the skis, he or she has to rent them for a fixed cost per day $c < B$, which we will choose w.l.o.g. to be equal to 1. 
 The challenge is to find an algorithm that minimizes the expected cost of the ski tour over the number of days $D$. 
 The optimal cost with foresight for this problem is clearly $\min(D, B)$, and the optimal online deterministic strategy has cost $2B - 1$ if $D \geq B$. So the deterministic competitive ratio is basically $2$. 
 It is known that one can do better by employing a randomized strategy \cite{karlin}: 

\begin{theorem}
\label{thm:skirental}
	Consider a strategy where, for each day $i \geq 1$,  with probability $p_i$, we rent skis on day $i$.
	Then, if we take 
	$ p (i) = \left( \frac{B - 1}{B} \right)^{B - i} \frac{1}{B[ 1 - ( 1 - 1 / B)^B] },$
	for days $i \leq B$, and $0$ otherwise, we pay expected total cost 
	$ \left(\frac{e}{e-1}\right) \min( D, B ).$
\end{theorem} 

Using additional knowledge about the adversarial distribution, e.g. the mean $\mu$, it is known that one can improve the strategy further \cite{constrained_skirental}:

\begin{theorem}
\label{thm:skirental-mean}
	Knowledge of the mean $\mu$ of the adversarial function $\pi(y)$ yields a new randomized strategy $p(x)$ for the standard ski rental problem, namely if $\frac{\mu}{B} < 2 \frac{e-2}{e-1}$, then 
	$$p(x)  = \begin{cases}
	\frac{1}{B(e-2)} e^{\frac{x}{B}}-1, &0 \leq x \leq B \\
	0, &\text{else.}
	\end{cases}$$ 
 The competitive ratio improves to
	$ 1 + \frac{\mu}{2B ( e - 2 )}$. Otherwise, the previous strategy is optimal. 
\end{theorem}

\section{The Transactional Conflict \\ Problem}
\label{sec:def}

We now formalize our problem, assuming the transactional model in Section~\ref{sec:model}. 
As the cost metric is different depending on the contention resolution strategy, we define the problem differently for each case. 
For simplicity, we will start by defining the problem for the case with two conflicting transactions, and then extend to the case with longer conflict chains. 
%
%
%

	\subsection{Problem Statement for Requestor Wins} \label{sec:req_wins} 


\paragraph{Two Conflicting Transactions} 
Assume that an executing transaction $T_1$ (the receiver) is interrupted by a transaction $T_2$ (the requestor), and that we have the choice between aborting $T_1$ immediately, and postponing  $T_1$'s abort by $x$ steps, in the hope that it commits. 
Assume that $T_1$ still has to execute for $D$ time steps (unknown) to commit, and that aborting $T_1$ immediately incurs a fixed cost $B$.\footnote{In practice, the cost $B$ will consist of the time for which the transaction $T_1$ has already been running when interrupted, plus a fixed non-trivial cleanup cost.}
%
%
We define the \emph{conflict cost} as follows: 
 (1) \textbf{If $D \leq x$}, then transaction $T_1$ commits at or before $x$, and we only pay the time by which we delayed $T_2$, which is $D$. In terms of the sum of running times, our delay decision added $D$ to the total cost. 
 (2) \textbf{If $D > x$}, then transaction  $T_1$ has not committed at $x$, and we will abort $T_1$, allowing $T_2$ to continue. 
We will pay the abort cost $B$, the additional time $x$ by which we ran $T_1$, and the additional time $x$ by which $T_2$ was delayed. This sums up to $2x + B$. 

\paragraph{The General Case} 
In general, the conflict chain may be longer, as several transactions may be delayed if we decide to extend $T_1$'s execution. For example, a third transaction $T_3$ may already be waiting on $T_2$ to commit at the point where $T_2$ conflicts with $T_1$. 
Assume that $k \geq 2$ transactions, including $T_1$, are conflicting, forming a conflict chain. This means that delaying $T_1$'s abort by \emph{one} step will add $k - 1$ time steps to the total running time of all transactions. 
In this case, the cost metric becomes:

\begin{itemize}
	\item If $D \leq x$, then $T_1$ commits, and we pay the time by which we delayed all transactions other than $T_1$, which is $(k - 1)D$. 
	\item If $D > x$, then we abort $T_1$, and pay the abort cost $B$ and the additional time $x$ by which we ran $T_1$, as well as the additional time $(k - 1)x$ by which other transactions are delayed. This sums up to $kx + B$. 
\end{itemize}

We will first  focus on decision algorithms which optimize the expected cost \emph{of each conflict}. 
Later, in Section~\ref{sec:competitive}, we will show how to use competitive bounds on individual conflict cost to obtain competitive bounds on the throughput. 
We are trying to identify the probability distribution ${p}$ which minimizes the expected decision cost under arbitrary adversarial choices $D$. 
Formally, for arbitrary $D$, and fixed $k$ and $B$ as above, 
\begin{align*}
	p & = \text{argmin}_{p} \text{Cost}(p(x),D), \text{ where} \\ 
	\text{Cost}(p(x),D) &= \int_{0}^{D}(kx + B) p(x) dx  \\ & + (k-1) \int_{D}^{\frac{B}{k-1}} D p(x) dx,
\end{align*}

\noindent where we have noticed that for $D > B / (k - 1)$ the algorithm will always abort. 
Notice that the offline optimum is $\min (B, (k - 1) D ).$ 

%


\subsection{Problem Statement for Requestor Aborts} \label{sec:req_aborts} 

\paragraph{Two Conflicting Transactions}  
The converse case is when the receiving transaction $T_1$ can abort or delay any requestor transaction $T_2$. 
Again, we might choose to delay the conflict resolution by a grace period of $x$ steps. 
Assume that $T_2$ still has to execute for $D$ steps to commit, where $D$ is unknown. 
The conflict cost is now as follows: 

\begin{itemize}
	\item If $D \leq x$, then transaction $T_1$ commits at or before $x$, and we only pay the time by which we delayed the incoming transaction $T_2$, i.e., $D$. 
	\item If $D > x$, then transaction  $T_1$ has not committed at $x$, and we will abort $T_2$ and continue. In this case, we pay as extra cost $B$, a fixed abort cost, plus the time by which $T_2$ was delayed, i.e. total cost $x + B$. 
\end{itemize}

Notice that the optimal cost with foresight is $\mbox{OPT} = \min ( B, D )$. By comparing variables we realize that there exists a direct mapping between the ski rental problem and the requestor aborts version of the transactional conflict problem:

The point where the requesting transaction $T_2$ interrupts transaction $T_1$ marks day 1 of the ski rental problem. Furthermore, we have that $D$, the unknown number of steps until $T_1$ would eventually commit, denotes the day on which the adversary chooses to stop us from skiing. The choice of $x$ i.e. the length of the grace period by which $T_1$ delays $T_2$, is the same as to choose that we buy skis on day $x+1$. Finally, the fixed extra cost $B$ in the transactional conflict problem is the equivalent to the cost $B$ of buying the skis. Note that in this case we assume that for the transactional conflict problem if $x=D$, $T_1$ is not able to commit and thus it aborts. If $x > D$, $T_1$ will be able to commit on time step $D$. 
Figure~\ref{fig:req_aborts} illustrates this case.
The results in Theorems~\ref{thm:skirental} and~\ref{thm:skirental-mean} apply to this case.

Generalizing to conflicts of size $k > 2$ is also possible:

\begin{theorem}	
	 if $\frac{\mu + 2(( k - 1 ) (e^{\frac{1}{k-1}} - 1 ) - 1)}{B} < 2(( k - 1 ) (e^{\frac{1}{k-1}} - 1 ) - 1)$ the optimal PDF is:
	\begin{align*}
	p(x) = \begin{cases}
	\frac{( k - 1 )}{B (( e^{\frac{1}{k-1}} - 1 ) ( k - 1 ) - 1)} ( e^{\frac{x}{B}} - 1 ), &0 \leq x \leq \frac{B}{k-1} \\
	0, &\text{otherwise}
	\end{cases},
	\end{align*}
and otherwise the optimal PDF is:
	\begin{align*}
	p(x) = \begin{cases}
	\frac{( k - 1 )}{B ( e^{\frac{1}{k-1}} -1) ( k - 1 )} e^{\frac{x}{B}}, &0 \leq x \leq \frac{B}{k-1} \\
	0, &\text{otherwise}
	\end{cases}
	\end{align*}
\end{theorem}

\begin{proof}
In the requestor aborts case for $k>2$ transactions,  we have one receiver transaction $T_1$ and $k-1$ requestor transactions.  Therefore in abort case, since we abort $k-1$ transactions, the extra cost becomes  $(k-1)(\Delta+B)$. This gives us the following cost function :
	
	\begin{align}
		\text{Cost(p(x),y)} = \int_{0}^{y} ( k - 1 )( x + B ) p(x) dx + ( k - 1 ) \int_{y}^{\frac{B}{k-1}} y p(x) dx		
	\end{align}
	
	Thus our Lagrangian function looks as follows:
	
	\begin{align*}
	L(p(x),\lambda_1,& \lambda_2) = \int_{0}^{\frac{B}{k-1}} \left( \frac{\text{Cost(p(x),y)}}{(k-1)y}  -\lambda_1 - \lambda_2 y \right)\pi(y) dy \\
	&+ \pi_K \cdot \left( \int_{0}^{\frac{B}{k-1}} \frac{( k - 1 ) ( x + B )}{B} p(x) dx - \lambda_1 - \lambda_2K\right) \\
	& + \lambda_1 + \lambda_2 \mu
	\end{align*}
	
	We note that, since the adversarial strategy is arbitrary, we will need the following two constraints:
	
	\begin{align}
	\frac{\text{Cost(p(x),y)}}{(k-1)y} &= \lambda_1 + \lambda_2 y \label{constr_ap_1} \\
	\int_{0}^{\frac{B}{k-1}} \frac{( k - 1 ) ( x + B )}{B} p(x) dx &= \lambda_1 + \lambda_2K \label{constr_ap_2}
	\end{align}
	
	By differentiating the constraint \ref{constr_ap_1} twice w.r.t. $y$ and substituting $y$ with $x$ we get
	
	\begin{align*}
	p'(x) B -  p(x)&= 2 \lambda_2 \\
	\end{align*}
	
	Solving this first order differential equation gives us:
	
	\begin{align}
		p(x) = \alpha e^{\frac{x}{B}} - 2 \lambda_2 \label{strategy_ap}
	\end{align}
	
	We can solve for $\alpha$ by using the fact that $p(x)$ is a PDF
	
	\begin{align}
		\alpha = \frac{( k - 1 ) + 2 B \lambda_2}{B ( e^{\frac{1}{k-1}} -1) ( k - 1 )} \label{alpha_ap}
	\end{align}
	
	By substituting \ref{strategy_ap} into our constraint \ref{constr_ap_1} we get
	
	\begin{align*}
		( \alpha B e^{\frac{y}{B}} &- 2 \lambda_2 y ) B - \alpha B^2 - \alpha B^2 e^{\frac{y}{B}} +  \lambda_2 y^2 + \\ \alpha B^2 &+ y \alpha B e^{\frac{1}{k-1}} - 2 y \lambda_2 \frac{B}{k-1} = y ( \lambda_1 + \lambda_2 y) \\
		\alpha ( &-B^2 + B^2 + B^2e^{\frac{y}{B}} - B^2e^{\frac{y}{B}} + y B e^{\frac{1}{k-1}} ) - \\  &\lambda_2 ( 2 y B - y^2 + 2 y \frac{B}{k-1} ) = y ( \lambda_1 + \lambda_2 y) \\
		\alpha B e^{\frac{1}{k-1}} &- \lambda_2 ( 2 B + 2 \frac{B}{k-1}) = \lambda_1
	\end{align*}
	
	Using \ref{alpha_ap} we get:
	
	\begin{align}
		\frac{e^{\frac{1}{k-1}}}{e^{\frac{1}{k-1}} - 1} - \lambda_2 \cdot 2 \left ( B - \frac{B}{ ( e^{\frac{1}{k-1}} - 1 ) ( k - 1 )} \right)&= \lambda_1 \label{constr1_solved_ap}
	\end{align}
	
	Similarly substituting \ref{strategy_ap} into constraint \ref{constr_ap_2} we get:
	
	\begin{align}
		\alpha B e^{\frac{1}{k-1}} &- \lambda_2 \left( 2 B + \frac{B}{k-1} \right) = \lambda_1 + \lambda_2 K \nonumber \\
		\frac{e^{\frac{1}{k-1}}}{e^{\frac{1}{k-1}} - 1} &- \nonumber \\ \lambda_2 & \cdot 2 \left ( \frac{B}{2 ( k - 1 )} + B + \frac{K}{2}  - \frac{B e^{\frac{1}{k-1}}}{ ( e^{\frac{1}{k-1}} - 1 ) ( k - 1 )}\right) = \lambda_1 \label{constr2_solved_ap}
	\end{align}
	
	In order for \ref{constr1_solved_ap} and \ref{constr2_solved_ap} to be true at the same time we need
	
	\begin{align}
		K &= \frac{B}{k-1} \label{constr_l_ap}
	\end{align}
	
	Since we know that $p(x)$ is a PDF we know that $\forall x \in [0, \frac{B}{k-1}], p(x) \geq 0$ and thus we get
	
	\begin{align*}
		\left(1 - \frac{e^{\frac{x}{B}}}{( e^{\frac{1}{k-1}} - 1 ) ( k - 1 )} \right)\lambda_2 &\leq \frac{e^{\frac{x}{B}}}{2 B ( e^{\frac{1}{k-1}} - 1 )}
	\end{align*}
	
	Thus we get the following constraints on $\lambda_2$:
	
	\begin{align}
		\lambda_2 &\leq \frac{ ( k - 1 ) e^{\frac{x}{B}}}{2 B (( k - 1 ) (e^{\frac{1}{k-1}} - 1 ) - e^{\frac{x}{B}})}, \nonumber\\ &\text{  if } 0 \leq x \leq B \log{\left(( k - 1 ) (e^{\frac{1}{k-1}} - 1 )\right)} \label{lambda2_1_ap}\\
		\lambda_2 &> \frac{ ( k - 1 ) e^{\frac{x}{B}}}{2 B (( k - 1 ) (e^{\frac{1}{k-1}} - 1 ) - e^{\frac{x}{B}})}, \nonumber \\ &\text{  if } B \log{\left(( k - 1 ) (e^{\frac{1}{k-1}} - 1 )\right)} < x \leq \frac{B}{k-1} \label{lambda2_2_Ap}
	\end{align}
	
	We immediately realize that the r.h.s. of \ref{lambda2_1_ap} is positive and strictly increasing, whereas the r.h.s. of \ref{lambda2_2_Ap} is negative. Thus we now know that we must have
	
	\begin{align*}
		0 \leq \lambda_2 \leq \frac{k-1}{2 B (( k - 1 ) (e^{\frac{1}{k-1}} - 1 ) - 1)}
	\end{align*}
	
	Thus our problem becomes:
	
	\begin{align}
		\min_{\lambda_1, \lambda_2} \lambda_1 + \lambda_2\mu \text{ s.t. } \nonumber \\
		\frac{e^{\frac{1}{k-1}}}{e^{\frac{1}{k-1}} - 1} - \lambda_2 \cdot 2 \left ( B - \frac{B e^{\frac{1}{k-1}}}{ ( e^{\frac{1}{k-1}} - 1 ) ( k - 1 )} + B \right)&= \lambda_1 \\
		\frac{e^{\frac{1}{k-1}}}{e^{\frac{1}{k-1}} - 1} - \lambda_2 \cdot 2 \left ( \frac{B}{2 ( k - 1 )} + B + \frac{K}{2}  - \frac{B e^{\frac{1}{k-1}}}{ ( e^{\frac{1}{k-1}} - 1 ) ( k - 1 )}\right) &= \lambda_1 \\
		\forall \lambda_1  \geq 0, 0 \leq \lambda_2 \leq \frac{k-1}{2 B (( k - 1 ) (e^{\frac{1}{k-1}} - 1 ) - 1)} \nonumber
	\end{align}
	
	We already realized that \ref{constr_l_ap}. We further remind ourselves that the solutions to this linear program have to form a convex polytope and that each basic feasible solution is a corner point of that polytope. Therefore we get the following two corner points
	
	\begin{align*}
		c_1 = \left(\frac{e^{\frac{1}{k-1}}}{e^{\frac{1}{k-1}} - 1},0\right) \text{ and } c_2 = \left(1,\frac{k - 1 }{2 B (( k - 1 ) (e^{\frac{1}{k-1}} - 1 ) - 1)}\right)
	\end{align*}
	
	This gives us the following $C$'s
	
	\begin{align*}
		C_1 = \frac{e^{\frac{1}{k-1}}}{e^{\frac{1}{k-1}} - 1} \text{ and } C_2 = 1 + \frac{\mu ( k - 1 )}{2 B (( k - 1 ) (e^{\frac{1}{k-1}} - 1 ) - 1)}
	\end{align*}
	
	Notice that for $\lambda_2 = 0$ we do not take the mean into account and thus are left with the solution for the non-constrained ski rental problem. For $\lambda_2 \neq 0$ we realize that $C_2 < C_1$ is only valid for small values of $\mu$. More precisely, $C_2 < C_1, \text{ for }  \frac{\mu + 2(( k - 1 ) (e^{\frac{1}{k-1}} - 1 ) - 1)}{B} < 2(( k - 1 ) (e^{\frac{1}{k-1}} - 1 ) - 1)$. For $B > 1$ it simplifies to $\frac{\mu}{B-1} < 2(( k - 1 ) (e^{\frac{1}{k-1}} - 1 ) - 1)$. 
\end{proof}

\section{Analysis for Requestor Wins} \label{sec:non_const}

In the following, we will focus on analyzing optimal strategies for the \emph{requestor wins} strategy. In particular, we will examine the optimal deterministic strategy for the \emph{unconstrained case}, in which no additional information about the adversarial distribution is known, and optimal \emph{randomized strategies} for the \emph{constrained case}, in which the first moment of the adversarial length distribution is known. (This case also convers the randomized unconstrained case, so we solve them together.) 

\subsection{The Deterministic Unconstrained Case}
Since the algorithm is deterministic, it has to choose a time step $x$ at which to abort. 
Denote by $D$ the time step at which the transaction would commit, if allowed to execute. 
Throughout the rest of this section, unless stated otherwise, the expression \textit{we abort} means that the receiving transaction aborts. 

Observe that the optimal cost is $\min((D(k-1),B)$.
For simplicity, let us  assume that $\frac{B}{k-1}$ is an integer. There are two cases we need to examine. The first one is $x \leq D$, in which we pay cost $kx +  B$. The second one is $x > D$,  in which we pay cost $(k-1)D$. 
Note that the adversary knows after which time step $x$ we decided to abort, and thus will never choose to set the end of transaction after $x$. 
The following result characterizes the optimal deterministic strategy. 

\begin{theorem} \label{Det-Non-Constr}
	The optimal deterministic strategy always chooses to abort after $\frac{B}{k-1}$ time steps. This strategy has total cost
	$ \left(2+\frac{1}{k-1} \right)\min((D(k-1),B).$
\end{theorem}
\begin{proof}
	Let $x$ denote the time step on which we chose to abort. 
	We first notice that we can reduce the case where $x \leq D$ to $x = D$, since our cost will not increase after the abort. 
	Thus, delaying the end of transaction can only decrease the competitive ratio. 
	The competitive ratio for $x = D$ looks as follows
	\begin{align}
	\frac{kx + B}{\min((k-1)x,B)} \label{det_first}
	\end{align}
	For $x > D$ we get 
	\begin{align*}
		\frac{D(k-1)}{\min(D(k-1),B)}
	\end{align*}
	as the competitive ratio. This clearly is a non-decreasing function. 
	Thus the adversary would choose $D = x-1$ to maximize our cost. 
	Therefore, we have competitive ratio $\frac{(x-1)(k-1)}{\min((x-1)(k-1),B)}$. Because this is at most $\frac{kx+B}{\min((k-1)x,B)}$, the adversary will always prefer $D = x$. 
	We obtain that our strategy will always yield competitive ratio \ref{det_first}. Analyzing this we get
	\begin{align*}
		\frac{kx + B}{\min((k-1)x,B)} = & 1 + \frac{x}{\min((k-1)x,B)} + \\  \frac{\max((k-1)x,B)}{\min((k-1)x,B)} , \nonumber \\ 
	\end{align*} 
	which is minimized for $x=\frac{B}{k-1}$, yielding  ratio: 
	\begin{align}
		2+\frac{1}{k-1}.		
		\label{det_ratio} 
	\end{align}

\end{proof}


\subsection{Randomized Strategies for Transactional Conflict} 
\label{sec:constrained}

In this section, we discuss the transactional conflict problem where we have some knowledge about the distribution $\pi(y)$ of the adversarial function. 
For simplicity, we will consider the case where $k = 2$, i.e. the conflict involves only two transactions, and address the general case in a later section.  
Our goal is to prove the following. 

\begin{theorem} \label{rqw::strategy}
	For arbitrary adversarial distributions, the following randomized strategy $p(x)$ is optimal: $p(x) = 1 / B \textnormal{ for }$ \\ $0 \leq x \leq B, \textnormal{ and } 0,$otherwise. This yields competitive ratio $2$ for the conflict cost.

	Knowledge of the mean $\mu$ of the adversarial function $\pi(y)$ yields the following strategy $p(x)$ for the transactional conflict problem.  
	If $\frac{\mu}{B} < 2(\ln 4 -1)$, then the optimal strategy is 
	$$p(x)  = \begin{cases}
	\frac{\ln\frac{B+x}{x}}{B(\ln 4 -1)}, &0 \leq x \leq B \\
	0, &\text{otherwise}.
	\end{cases}$$ 
	In this case, the competitive ratio improves to
	$ 1 + \frac{\mu}{2B (\ln 4 - 1)}. $ 
	If $\frac{\mu}{B} \geq 2(\ln 4 -1),$ then the unconstrained strategy is optimal. 
	
\end{theorem}
\begin{proof} 
It is easy to show that for arbitrary adversarial distributions above randomized strategy yields competitive ratio 2. Optimality of this strategy can be proven using the same techniques as in \cite{constrained_skirental}. Note that for $k>2$, the following strategy is optimal and 2-competitive:  
$p(x) = k-1 / B \textnormal{ for }$ \\ $0 \leq x \leq B/(k-1), \textnormal{ and } 0,$otherwise. 
Proof is analogous to $k=2$ case.

Let $x$ be a time step at which we decide to abort and let $y$ be a time step at which we would finish the computation.
$x$ is chosen with the distribution $p$ and $y$ is chosen with the distribution $\pi$. 
Notice that we will always abort at time step $B$ the latest, because otherwise our cost will be greater than if we had aborted at time step 0.  If $y < x$, we pay $y$. Otherwise, we pay $2x + B$. 
Let $Cost(p(x),y)$ be the expected cost for the fixed $y$. We have that
  \begin{align*}
	\text{Cost}(p(x),y) &= \int_{0}^{y}(2x+B)p(x) dx + y \cdot \int_{y}^{B}p(x)dx.
	\end{align*}
We know that the optimal cost is $\min(y,B)$. This gives us a necessity to introduce non-zero probability $\pi_K$ for $y$ to be chosen from outside the interval $[0,B]$, because otherwise the trivial optimal strategy is to never abort. Our goal is to minimize the ratio
	
\begin{align*}
	\mathcal{G}(p(x),\pi(y)) &= \int_{0}^{B}\frac{\text{Cost}(p(x),y)}{y}\pi(y)dy  \\ & + \pi_K \cdot \int_{0}^{B}\frac{2x + B}{B}p(x)dx.
\end{align*} 
Considering that we want to find the best possible solution under the worst case distribution $\pi(y)$ of the adversary our problem looks as follows
\begin{align}
	\min_{p(x)} \max_{\pi(y)} \mathcal{G}(p(x),\pi(y)). \label{max_prob}
\end{align}

We will give a closed-form solution for this as follows. We first construct the Lagrangian function $\Lagr$ for the maximization problem in \ref{max_prob}. Afterwards we will construct its dual to get the necessary constraints. We will then differentiate one of the constraints twice, leaving us with a first order differential equation for $p(x)$. Solving this yields the desired PDF $p(x)$. Finally, we substitute $p(x)$ into the constraints to solve for the Lagrangian multipliers. 

The Lagrangian function $\Lagr$ takes as inputs the probability distribution $\pi(y)$ of the adversary and the two Lagrange multipliers $\lambda_1$ and $\lambda_2$. The two multipliers act as weights for the constraints given by the fact that $
\pi(y)$ is a PDF with the mean $\mu$. Thus, we get the following Lagrangian for the maximization problem:

\begin{align*}
	\Lagr(\pi(y), \lambda_1, \lambda_2) &= \int_{0}^{B} \frac{\text{Cost}(p(x), y)}{y}\pi(y)dy  \\ &+\pi_K \cdot \int_{0}^{B} \frac{2x+B}{B}p(x)dx \\
	\\ &- \lambda_1 \cdot \left[\int_{0}^{B}\pi(y)dy + \pi_K - 1\right] \\ &- \lambda_2 \cdot \left[\int_{0}^{B}y\pi(y)dy + K \cdot \pi_K - \mu \right] \\ &= \int_{0}^{B} \left(\frac{\text{Cost}(p(x), y)}{y} - \lambda_1 - \lambda_2y\right) \cdot \pi(y)dy \\
	&+ \pi_K \cdot \left( \int_{0}^{B} \frac{2x+B}{B}p(x)dx - \lambda_1 - \lambda_2K\right) 
	\\ &+ \left(\lambda_1 + \lambda_2\mu\right).
\end{align*}
Therefore, its dual is 
$	h(\lambda_1, \lambda_2) = \sup_{\pi(y)} \Lagr(\pi(y),\lambda_1,\lambda_2).$
We note that, since the adversarial strategy is arbitrary and we want to minimize the dual function, we will need constraints
$\frac{\text{Cost}(p(x), y)}{y} - \lambda_1 - \lambda_2y = 0$
and
$\int_{0}^{B} \frac{2x+B}{B}p(x)dx - \lambda_1 - \lambda_2K = 0$, to ensure that the Lagrangian function does not depend on the choice of $\pi(y)$. Therefore  we get the following minimization problem:
\begin{align}
	\min_{p(x), \lambda_1, \lambda_2} \lambda_1 + \lambda_2\mu \text{ s.t. } \nonumber \\
	\frac{\int_{0}^{y}(2x+B)p(x) dx + y \cdot \int_{y}^{B}p(x)dx}{y} = \lambda_1 + \lambda_2y \label{constr_diff_1} \\
	\int_{0}^{B}\frac{2x+B}{B}p(x)dx = \lambda_1 + \lambda_2K \label{constr_diff_2} \\
	\forall \lambda_1, \lambda_2 \geq 0 \text{ and } y \in [0,B] \nonumber
\end{align}

Because  the constraint \ref{constr_diff_1} holds for every $y \in [0, B]$, we can differentiate it twice with respect to $y$ and after we substitute $y$ with $x$ we get the first order differential equation
$	p'(x) = \frac{1}{B+x}  2\lambda_2$.
Solving this yields our strategy $p(x) = \alpha + 2\lambda_2 \ln (B+x)$.
Using the fact that $p(x)$ is a PDF we get that $\alpha = \frac{1}{B} - 2\lambda_2(\ln B + \ln 4 -1)$.
Since $\forall x, p(x) \geq 0$ and $\lambda_2 \geq 0$ we further have 
$0 \leq \lambda_2 \leq \frac{1}{2B (\ln 4 - 1)}$. Thus our problem becomes:

\begin{align}
	\min_{\lambda_1, \lambda_2} \lambda_1 + \lambda_2\mu \text{ such that } \nonumber \\
	\left[(4\ln 2B - 4 \ln 4B + 2)B\right] \lambda_2 + 2 &= \lambda_1 \label{constr_final_1} \\
	\left[(4\ln 2B - 4 \ln 4B + 3)B - K\right] \lambda_2 + 2 &= \lambda_1 \label{constr_final_2} \\
	\forall \lambda_1  \geq 0, 0 \leq \lambda_2 \leq \frac{1}{2B (\ln 4 - 1)}  \nonumber
\end{align}
We immediately realize that for $\lambda_2 \neq 0$ we must have $K = B$ in order to satisfy \vref{constr_final_1} and \ref{constr_final_2} simultaneously. We remind ourselves that the solutions to this linear program have to form a convex polytope and that each basic feasible solution is a corner point of that polytope. Therefore we get two corner points
$	c_1 = \left(2,0\right) \text{ and } c_2 = \left(1,\frac{1}{2B (\ln 4 - 1)}\right)
$
with corresponding competitive ratios
$	C_1 = 2 \text{ and } C_2 = 1 + \frac{\mu}{2B (\ln 4 - 1)}$.
Notice that for $\lambda_2 = 0$ we do not take the mean into account and thus are left with the solution for the non-constrained ski rental problem. For $\lambda_2 \neq 0$ we realize that $C_2 < C_1$ is only valid for small values of $\mu$. More precisely, $C_2 < C_1, \text{ for }  \frac{\mu}{B} < 2(\ln 4 -1)$.  Thus, 
if $\frac{\mu}{B} < 2(\ln 4 -1)$ the optimal PDF is:
\begin{align*}
	p(x) = \begin{cases}
		\frac{\ln\frac{B+x}{B}}{B(\ln 4 -1)}, &0 \leq x \leq B \\
		0, &\text{otherwise}.
	\end{cases}
\end{align*}
Otherwise, the unconstrained strategy is optimal.
\end{proof}

\subsection{Discussion}

\paragraph{Competitive Ratio Comparison} We compare this result with the requestor aborts case (classic ski rental)~\cite{constrained_skirental}. Depending on the case distinction, we get:
\begin{itemize}
	\item inequality holds: In this case we have a competitive ratio of $1 + \frac{\mu}{2B (\ln 4 - 1)}$ for requestor wins and $1 + \frac{\mu}{2B ( e-2 )}$ for requestor aborts. Clearly, requestor aborts outperforms requestor wins. Additionally, we get that the inequality which has to hold in order for this strategy to be applicable, is less strict for the requestor aborts case.
		
	\item inequality does not hold: If the inequality does not hold we get competitive ratio $2$ for the requestor wins strategy and $\frac{e}{e-1}$ for the requestor aborts strategy. Again, we notice that the competitive ratio of requestor aborts is smaller than the one for requestor wins.
\end{itemize}

Thus, we can conclude that with respect to the competitive ratio for a single waiting transaction ($k = 2$), the transactional conflict problem should be tackled by the requestor aborts strategy. We note that this may no longer be the case for $k \geq 3$. 

\paragraph{Abort probability} We now examine the probability that a transaction aborts, in the case where $y \leq B$. (Otherwise, it is more advantageous to abort the current transaction than to delay the system for $> B$ steps.) 
Notice that in this case, the adversary's best strategy is to set $y = B$. 
We obtain the following by direct computation:
\begin{itemize}
	\item requestor wins: $1-p (B) = 1 - \frac{\ln 2}{B ( \ln 4 - 1 )} \simeq 1 - 1.8 / B$
	\item requestor aborts: $1-p(B) = 1 - \frac{e-1}{B(e-2)} \simeq 1 - 2.4 / B.$
\end{itemize}
Hence, the requestor aborts optimal strategy is less likely to abort a transaction, under the same conditions.


\subsection{Constrained Problem for \\ Conflict Size $k>2$ }

A more involved analysis solves the general case where conflicts are of size $k > 2$.
\begin{theorem} \label{constrkmore2}
Consider an instance of requestor-wins transactional conflict where $k \geq 3$ transactions are involved. 
Given the fixed cost $B >0$, and the mean of the adversarial distribution $\mu > 0$, then the optimal PDF $p(x)$ is as follows. 

\noindent If $\frac{\mu}{B} \le \frac{k^{k-1}-2(k-1)^{k-1}}{(k-2)\left(k^{k-1}-(k-1)^{k-1}\right)}$, then :
$$
p(x)=
\begin{cases}
\frac{(B+x)^{k-2}(k-1)^k\left( 2(k-1)^{k-1}+k^{k-1}\right)}{B^{k-1}\left(k^{k-1}-(k-1)^{k-1}\right)\left(k^{k-1}-2(k-1)^{k-1}\right)}-\\ \frac{4(k-1)^k}{B\left(k^{k-1}-2(k-1)^{k-1}\right)},   0\le x \le \frac{B}{k-1},\\
0,  \text{otherwise}.\\
\end{cases}
$$
Otherwise, or if the mean is unknown, the optimal PDF is :
$$
p(x)=
\begin{cases}
\frac{(B+x)^{k-2}(k-1)^k}{B^{k-1}\left(k^{k-1}-(k-1)^{k-1}\right)} , & 0\le x \le \frac{B}{k-1} ,\\
0, & \text{otherwise} .\\
\end{cases}
$$
\end{theorem}
\begin{proof}
We can follow the same steps as in the $k=2$ case, but with more technical care. 
Our cost function becomes:
\begin{align*}
	\text{Cost}(p(x),y) &= \int_{0}^{y}(kx + B) p(x) dx + (k-1) \int_{y}^{\frac{B}{k-1}} y p(x) dx,
\end{align*}
and so the corresponding Lagrangian function is:
\begin{align*}
	L(p(x),\lambda_1, \lambda_2) = &\int_{0}^{\frac{B}{k-1}} \left( \frac{\text{Cost(p(x),y)}}{(k-1)y}  -\lambda_1 - \lambda_2 y \right)\pi(y) dy \\
	&+ \pi_K \cdot \left( \int_{0}^{\frac{B}{k-1}} \frac{kx+B}{B} p(x) dx - \lambda_1 - \lambda_2K \right) \\
	& + \lambda_1 + \lambda_2 \mu.
\end{align*}

We note that, since the adversarial strategy is arbitrary, we will need the following two constraints:
\begin{align}
	\frac{\text{Cost(p(x),y)}}{(k-1)y} &= \lambda_1 + \lambda_2 y \\ \label{Constr1}
	\int_{0}^{\frac{B}{k-1}} \frac{kx+B}{B} p(x) dx &= \lambda_1 + \lambda_2K \\ 
	\nonumber
\end{align}

which yields the linear program:
\begin{align}
	\min_{p(x), \lambda_1, \lambda_2} \lambda_1 + \lambda_2\mu \text{ s.t. } \nonumber \\
	\frac{\int_{0}^{y}(kx+B)p(x) dx + y \cdot \int_{y}^{B}p(x)dx}{(k-1)y} = \lambda_1 + \lambda_2y \label{Constr1} \\
	\int_{0}^{B}\frac{kx+B}{B}p(x)dx = \lambda_1 + \lambda_2K \label{Constr2} \\
	\forall \lambda_1, \lambda_2 \geq 0 \text{ and } y \in [0,\frac{B}{k-1}] \nonumber
\end{align}

By differentiating \ref{Constr1} twice with respect to $y$ and substituting $y$ with $x$ we get:
\begin{align*}
p'(x) ( x + B ) - ( k-2 ) p(x)&= 2 (k-1) \lambda_2.
\end{align*}

Solving this first order differential equation we get
\begin{align*}
	p(x) &= \alpha ( B+x )^{ k-2 } - \frac{2\lambda_2 ( k - 1 )}{k-2}.
\end{align*}

Using the fact that $p(x)$ is a PDF we can solve for $\alpha$:
\begin{align*}
	\alpha &= \frac{B ( k - 1 ) k ( 2 \lambda_2 B+ k - 2 )}{( k - 2 ) \left[ ( k - 1 ) ( \frac{ B k }{ k - 1 })^k - k B^k \right]} \\
	&= \frac{\lambda_2 ( 2 B^2 k ( k - 1 ))}{ (k - 2 ) \left[ ( k - 1 ) ( \frac{ B k }{ k - 1 })^k - k B^k \right]} + \frac{B ( k - 1 ) k}{  ( k - 1 ) ( \frac{ B k }{ k - 1 })^k - k B^k }.
\end{align*}


Substituting $p(x)$ into the constraint \ref{Constr1} yields:
\begin{align*}
\lambda_2 \left( \frac{2B}{k-2} \cdot \frac{k^{k-1}}{k^{k-1}-(k-1)^{k-1}}-\frac{4B}{k-2} \right) &+ \\ \frac{k^{k-1}}{k^{k-1}-(k-1)^{k-1}} &= \lambda_1.
\end{align*}

We know that $p(x)$ is a PDF so we have $p(x) \ge 0$ for $x \in [0, \frac{B}{k-1}]$. This means that $P(x) = \alpha ( B + x )^{ k-2 }- 
\frac{2\lambda_2(k-1)}{k-2} \ge 0$. Also, $\alpha > 0$ and $B>0$, and so $p(x)$ is an increasing function. Thus, it is enough to check if $p(0) \ge 0$. This gives us following constraint on $\lambda_2$:

\begin{align*}
 \lambda_2 \le \frac{2(k-2)}{B}\frac{\frac{(k-1)^{k-1}}{k^{k-1}-(k-1)^{k-1}}}{1-\frac{(k-1)^{k-1}}{k^{k-1}-(k-1)^{k-1}}} &= \frac{2(k-2)(k-1)^{k-1}}{B\left(k^{k-1}-2(k-1)^{k-1}\right)}
 \end{align*}
 
 So our problem becomes:
 
 \begin{align}
	\min_{\lambda_1, \lambda_2} \lambda_1 + \lambda_2\mu \text{ s.t. } \nonumber \\
\lambda_2 \left( \frac{2B}{k-2} \cdot \frac{k^{k-1}}{k^{k-1}-(k-1)^{k-1}}-\frac{4B}{k-2} \right) &+ \nonumber \\ \frac{k^{k-1}}{k^{k-1}-(k-1)^{k-1}} &= \lambda_1 \\ \label{Constr4}
  \lambda_2 \left( \frac{2B}{k-2} \cdot \frac{k^{k-1}}{k^{k-1}-(k-1)^{k-1}}-\frac{4B}{k-2} + 
  \frac{B}{k-1}-K \right) &+ \nonumber \\ \frac{k^{k-1}}{k^{k-1}-(k-1)^{k-1}} &= \lambda_1\\ \label{Constr5}
  \forall \lambda_1 \ge 0, 0\le  \lambda_2  \le \frac{2(k-2)(k-1)^{k-1}}{B\left(k^{k-1}-2(k-1)^{k-1}\right)} \nonumber
 \end{align}
 It is easy to see that we must have $K=\frac{B}{k-1}$.
 We also know that the solutions of LP form a convex polytope and this gives us two corner points : 
$$\left(\frac{k^{k-1}}{k^{k-1}-(k-1)^{k-1}}, 0\right)$$ and 
$$\left(\frac{k^{k-1}-2(k-1)^{k-1}}{k^{k-1}-(k-1)^{k-1}}, \frac{2(k-2)(k-1)^{k-1}}{B\left(k^{k-1}-2(k-1)^{k-1}\right)}\right).$$  
The corresponding ratios are : $\frac{k^{k-1}}{k^{k-1}-(k-1)^{k-1}}$ and \\ $\frac{k^{k-1}-2(k-1)^{k-1}}{k^{k-1}-(k-1)^{k-1}}+\frac{2\mu(k-2)(k-1)^{k-1}}{B\left(k^{k-1}-2(k-1)^{k-1}\right)}$.
\end{proof}

For  large $k$, using Theorem \ref{constrkmore2} and    $\left({\frac{k}{k-1}}\right)^{k-1} \simeq e$, we get that, if $\frac{\mu}{B} \le \frac{e-2}{(k-2)(e-1)}$, then the optimal PDF is:
$$
p(x)=
\begin{cases}
\frac{(B+x)^{k-2}(k-1)(2+e)}{B^{k-1}(e-1)(e-2)}-\frac{4(k-1)}{B(e-2)}, &\text{if } 0 \le x \le \frac{B}{k-1}\\
0, & \text{otherwise} .\\
\end{cases}
$$
and otherwise the optimal PDF is:
$$
p(x)=
\begin{cases}
\frac{(B+x)^{k-2}(k-1)}{B^{k-1}(e-1)}, &\text{if } 0\le x \le \frac{B}{k-1}\\
0, & \text{otherwise} .\\
\end{cases}
$$

\section{Competitive Analysis for the Sum of Running Times}
\label{sec:competitive}

%
Given an adversarial strategy $S$, and an algorithm $\mathcal{A}$, recall that $\Gamma(T, \mathcal{A})$ is the expected time between the start of the transaction's execution, and the time it committed. 
We define  the \emph{commit cost} $\rho_T$ of a transaction $T$ the number of consecutive time steps it has to execute for in isolation in order to commit. 

Given the optimal decision algorithm, a transaction $T$, and an adversarial strategy $S$, let the \emph{abort cost} $\alpha_T (S)$ under an adversarial strategy be the length of time 
that the algorithm spends executing $T$ \emph{minus} the commit cost $\rho_T$. 
Notice that, in the optimal algorithm, every time a transaction is interrupted by the adversary's strategy, the choice of whether to continue or not is deterministic, based on the length of time for which the transaction has already executed, on the abort cost, and on the remaining length of time for which the transaction has to execute.
Given an adversarial strategy $S$, define $w(S)$ as the \emph{waste} of the optimal algorithm given $S$, defined as $\sum_T \alpha_T(S) / \sum_T \rho_T$. 

\begin{corollary}
	Under the above conflict model, let $\mathcal{A}$ be the randomized requestor wins strategy, and let $OPT$ be the offline optimal algorithm. 
	Then we have that $$ \frac{\sum_T \Gamma(T, \mathcal{A}) }{\sum_T \Gamma(T, OPT) } \leq \frac{2 w(S) + 1}{w(S) + 1}.$$
\end{corollary}

\begin{proof}
Let $C$ be a conflict which arises for algorithm $\mathcal{A}$ following strategy $S$. From our conflict model, we know that the same conflict $C$ must arise for the optimal decision algorithm as well, although the decision may be different. 
 We know that there is one \emph{receiver} transaction $T$ involved in the conflict, to which we will \emph{amortize the cost of this conflict}. 
 For our algorithm $\mathcal{A}$, let $Cost(C, \mathcal{A})$ be the conflict cost, that is, the sum of delays caused by transaction $T$, plus the abort cost and delay to $T$ incurred in case $T$ does not commit. 

 Analogously, let $Cost(C, OPT)$ be the cost of a conflict  $C$
for the optimal decision algorithm. We get that the sum $\sum_C  Cost(C, OPT)=\sum_T \alpha_T(S)$.
Because of the way we defined $Cost(C, OPT)$ and $Cost(C, \mathcal{A})$ we know from the properties of the local decision algorithm that 
$$\frac{Cost(C, \mathcal{A})}{Cost(C, OPT)} \le 2.$$ 

\noindent Therefore, we also have that $$\frac{\sum_C  Cost(C, \mathcal{A})}{\sum_C Cost(C, OPT)} \le 2.$$ 
Finally, notice that:
\begin{align*}
 \frac{\sum_T \Gamma(T, \mathcal{A}) }{\sum_T \Gamma(T, OPT) } &=
\frac{\sum_T \rho_T + \sum_C Cost(C, \mathcal{A})}{\sum_T \rho_T + \sum_C Cost(C, OPT) }  \\ &\le \frac{\sum_T \rho_T + 2\sum_T \alpha_T}{\sum_T \rho_T + \sum_T \alpha_T}= \frac{2 w(S) + 1}{w(S) + 1},
 \end{align*}
\noindent which concludes the proof.
\end{proof}

\section{Throughput versus Progress} 
\label{sec:progress}

As described, our framework optimizes solely for throughput, and does not provide any progress guarantees. 
In particular, a transaction $T$ which consistently incurs conflicts at a time when its abort cost $B$ is smaller than its remaining execution time will always abort, since it is more advantageous for the system overall to abort this transaction than to delay. 
However, we can easily adapt our scheme in a backoff-like manner to address this issue: upon abort, a transaction can increase its future abort cost $B$ by an additive or multiplicative amount, and therefore be less likely to abort on its next execution. 
We consider the multiplicative case here, and obtain the following probabilistic progress guarantee for every transaction: 

\begin{corollary}
	\label{cor:progress}
	A transaction $T$ which encounters $\gamma$ conflicts during its execution and has  a running time $y$,  commits after at most $\log{y}+\log{\gamma}+\log{k}-\log{B}+2$ attemps, with probability at least $\frac{1}{2}$. 
\end{corollary}

\begin{proof}
We consider the non-constrained case, where we do not have any knowledge about $y$.
Notice that it is enough to prove that the bounds hold for requestor wins scenario, since
requestor aborts strategy is less likely to abort the transaction. 
As noted in Theorem \ref{rqw::strategy},  the optimal 2-competitive strategy in this case is:
$p(x) = k-1 / B \textnormal{ for }$ \\ $0 \leq x \leq B/(k-1), \textnormal{ and } 0,$.
Consider the situation after the transaction aborts $\log{y}+\log{\gamma}+\log{k}-\log{B}+1$ times.
Let $B'$ be the current abort cost. Since we double this cost  every time transaction aborts, we have that $B' \ge 2ky\gamma$.
Upon each conflict, the probability that we do not abort transaction is :
\begin{equation*}
\frac{\frac{B'}{k-1}-y}{\frac{B'}{k-1}}=1-\frac{y(k-1)}{B'} \ge 1-\frac{1}{2\gamma}.
\end{equation*}
Thus, the probability that transaction $T$ commits after  $\log{y}+\log{\gamma}+\log{k}-\log{B}+2$ attempts is at least:
\begin{equation*}
{(1-\frac{1}{2\gamma})}^{\gamma} \ge 1/2, 
\end{equation*}
\noindent where in the last step we used Bernoulli's inequality.

\end{proof}
\section{Experiments} \label{exp}

\begin{figure*}[t]
\centering
\begin{subfigure}[b]{0.31\textwidth}
\includegraphics[width=\textwidth]{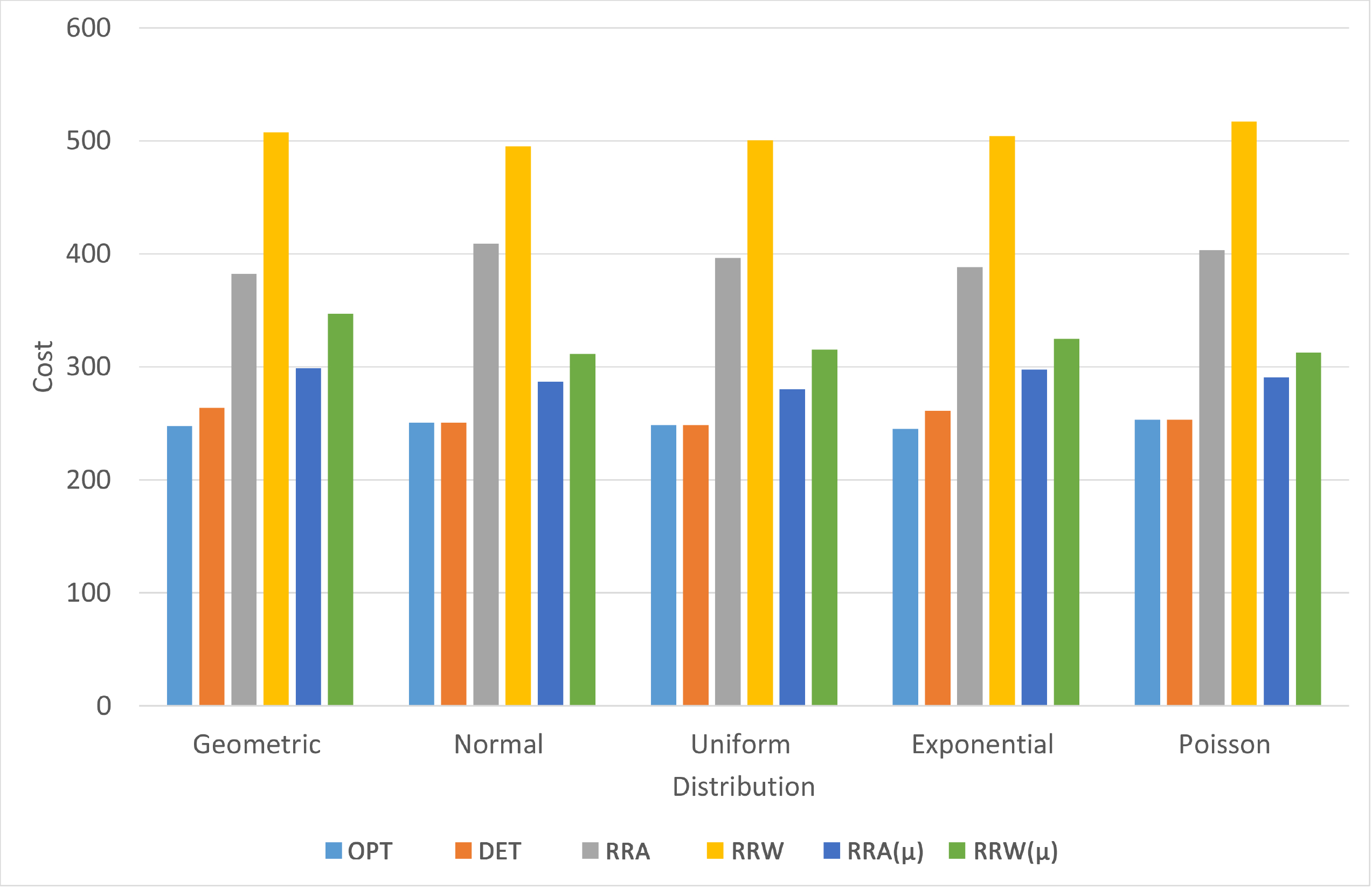}
\caption{Average cost of various strategies in the case with high fixed cost. Here, the fixed abort cost is $B = 2000$, and the mean is $\mu$ = 500.}
\label{fig:1}
\end{subfigure}
~
\begin{subfigure}[b]{0.31\textwidth}
\includegraphics[width=\textwidth]{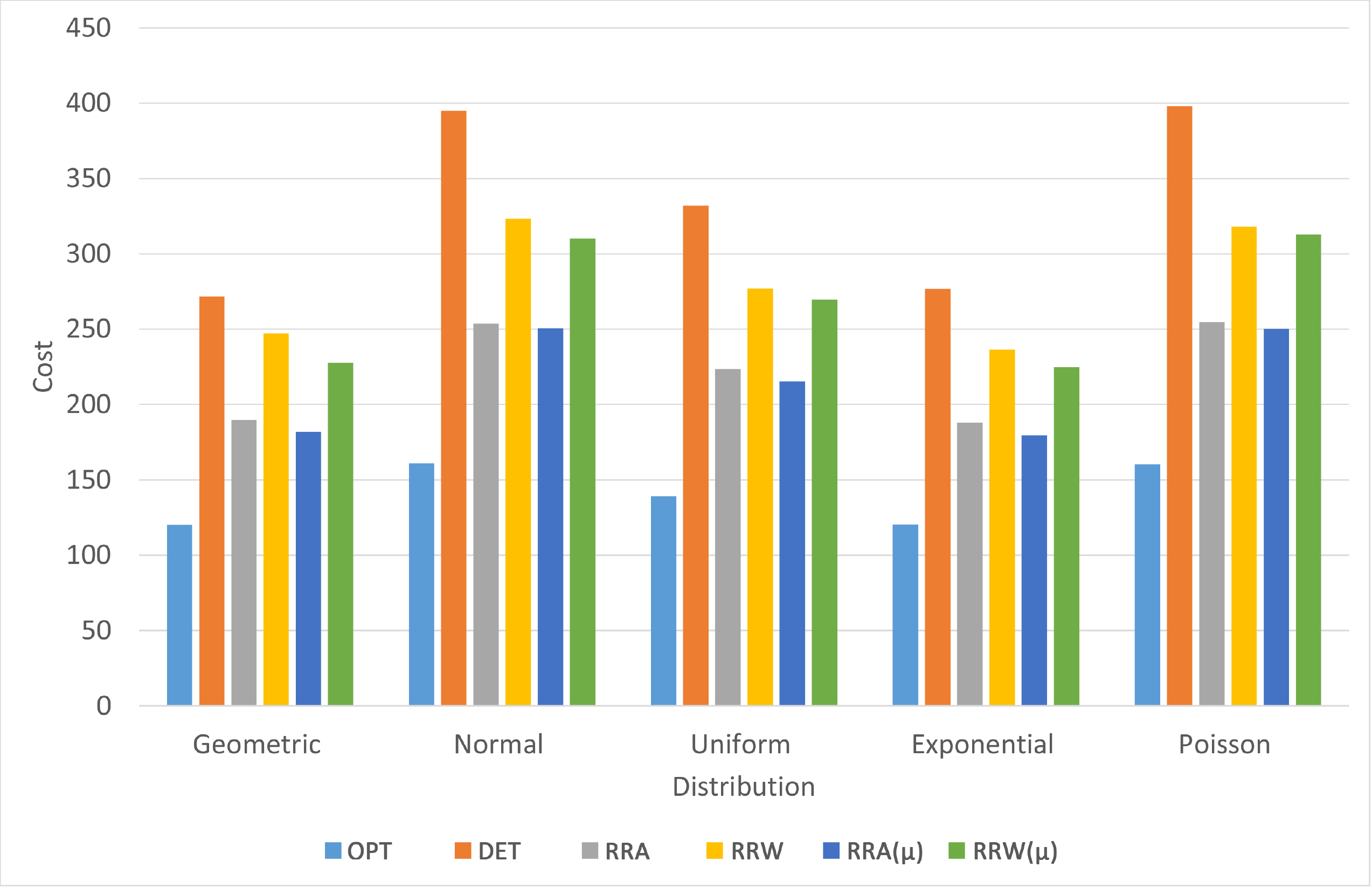}
\caption{Average cost of various strategies in the case with low fixed cost. Here, the fixed abort cost is $B = 200$ and the mean is $\mu$ = 500.}
\label{fig:2}
\end{subfigure}
~
\begin{subfigure}[b]{0.31\textwidth}
	\includegraphics[width=\textwidth]{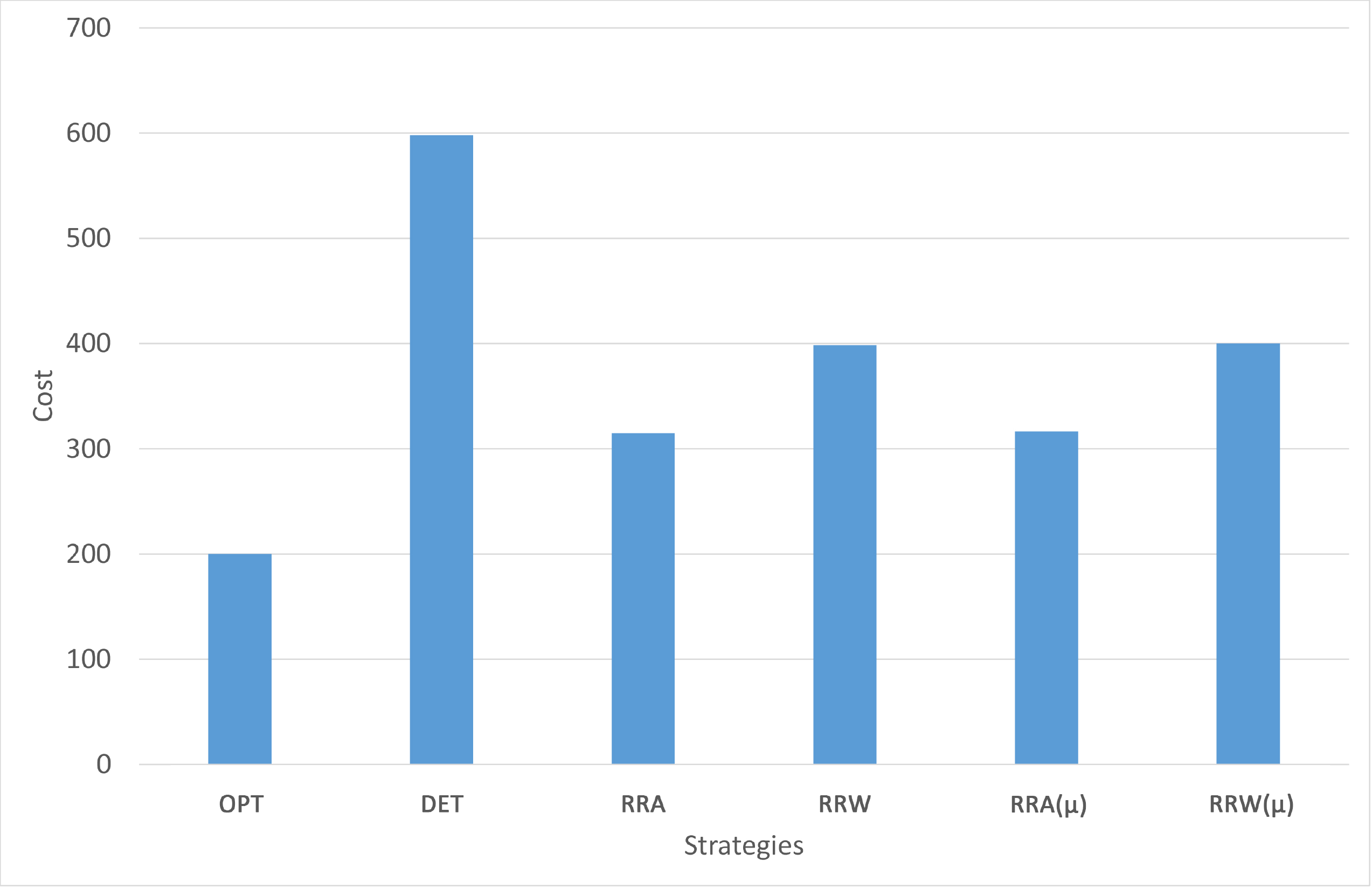}
	\caption{Average cost of various strategies if the adversary uses the  worst-case distribution for the deterministic strategy.}
	\label{fig:worst_case}
\end{subfigure}
\caption{Results of synthetic tests.} 
\end{figure*}

\subsection{Synthetic Tests}
In this section, we examine how various strategies perform for different types of distributions, in the case where two transactions are conflicting. We will use the following abbreviations: 
\begin{itemize}
	\item $RRW(\mu)$ denotes the randomized strategy that we derived for requestor wins using the constraint on the mean.
	\item $RRA(\mu)$ denotes the randomized strategy for requestor aborts using the constraint on the mean.
	\item $RRW$ denotes the randomized strategy for requestor wins we derived without using constraints.
	\item $RRA$ denotes the randomized strategy for requestor aborts without using constraints.
	\item $DET$ denotes the deterministic strategy which we also derived for requestor wins without using constraints.
	\item $OPT$ denotes the optimal strategy.
\end{itemize}

We benchmarked the decision algorithms in Python, as follows: First, we draw the length of the transaction $r$ from a given length distribution. 
We then pick an index $i$ u.a.r. from that length, which will be the point of interrupt in the ski rental problem. Note that $r - i$ is the equivalent to the mean in the ski rental problem. Then, the algorithm picks an index $j$ according to the above strategies, which will be the time at which we abort. Finally, we calculate the cost of this choice. 
The following length distributions were used in the experiment: Geometric, Normal, Uniform, Exponential and Poisson. The results are shown in Figures~\ref{fig:1} and~\ref{fig:2}.


\begin{figure*}[ht]
\centering
\begin{subfigure}[b]{0.45\textwidth}
\includegraphics[width=0.9\columnwidth]{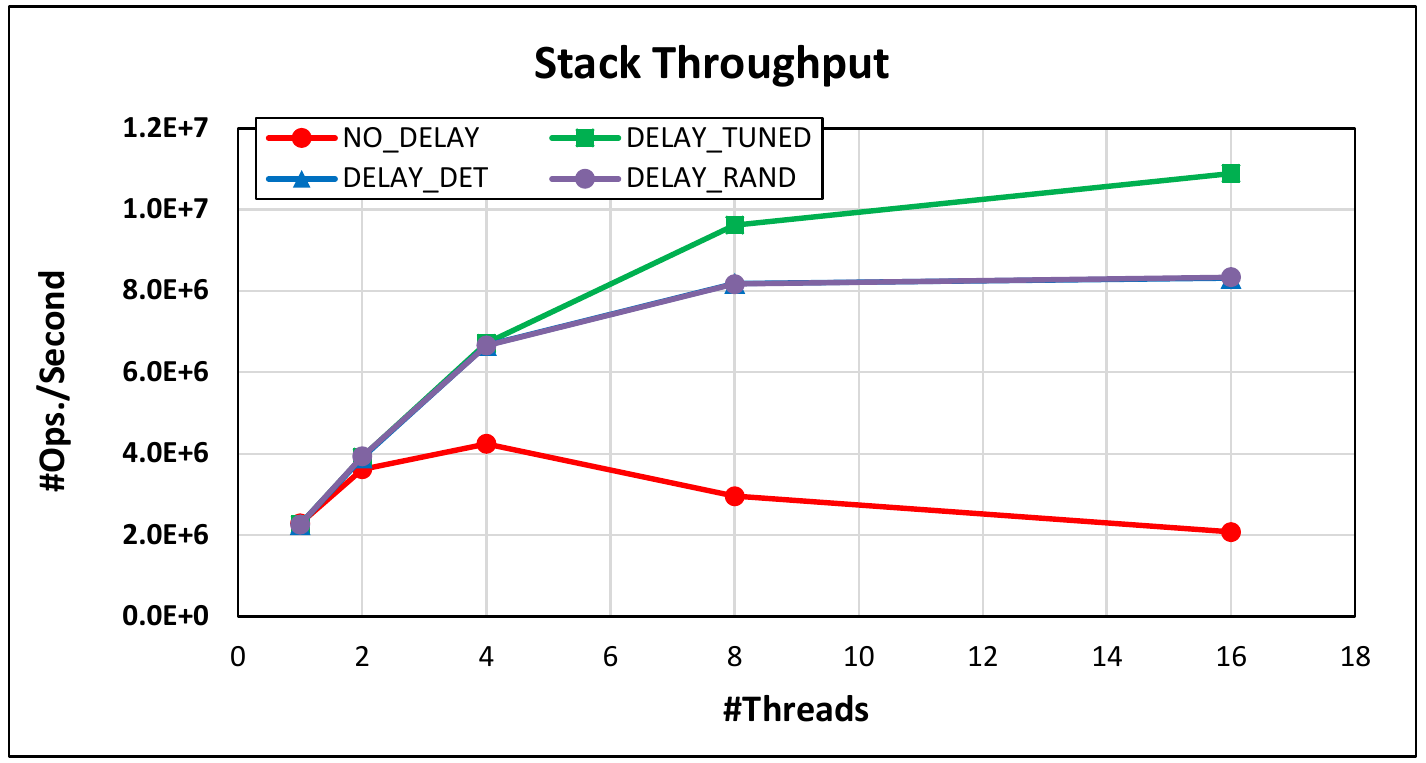}
\end{subfigure}
\begin{subfigure}[b]{0.45\textwidth}
\includegraphics[width=0.9\columnwidth]{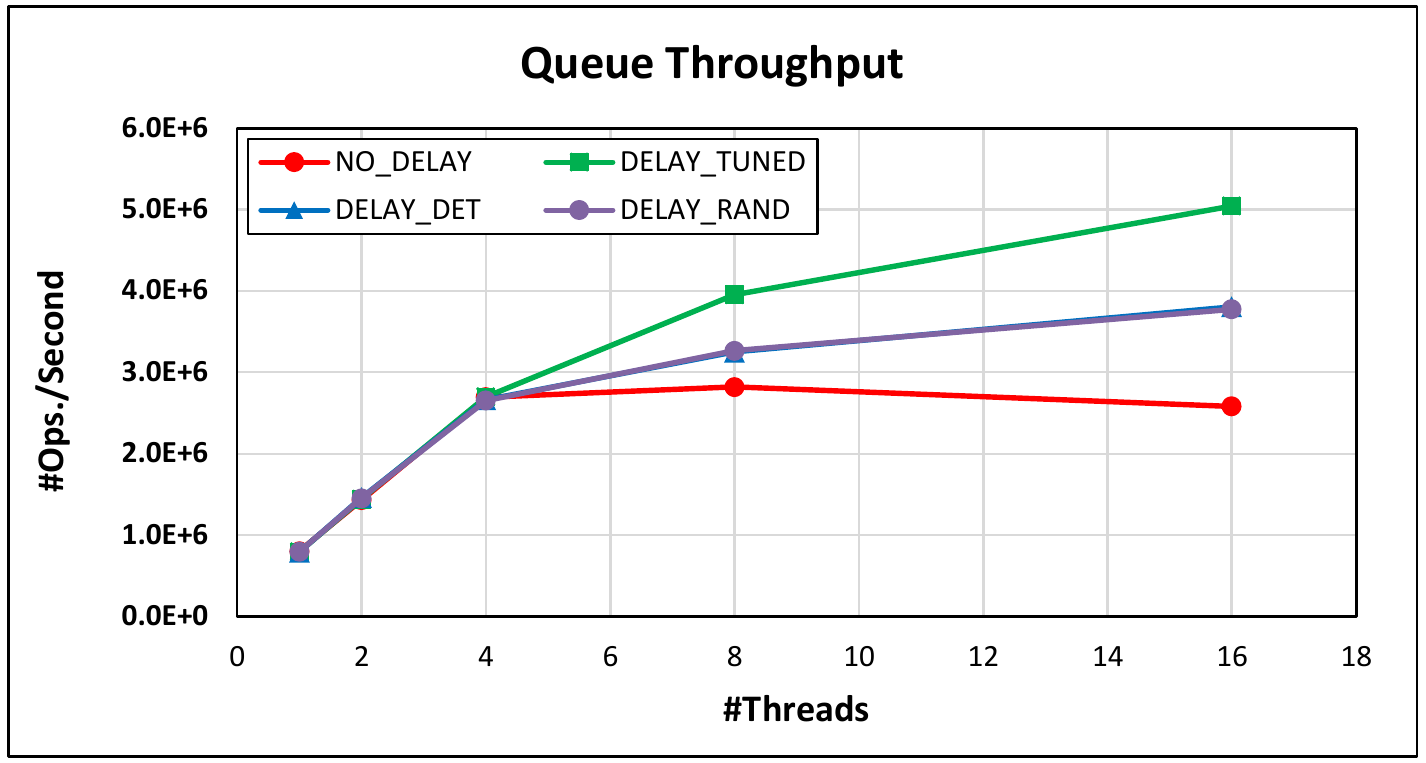}
\end{subfigure}
\\
\begin{subfigure}[b]{0.45\textwidth}
\includegraphics[width=0.9\columnwidth]{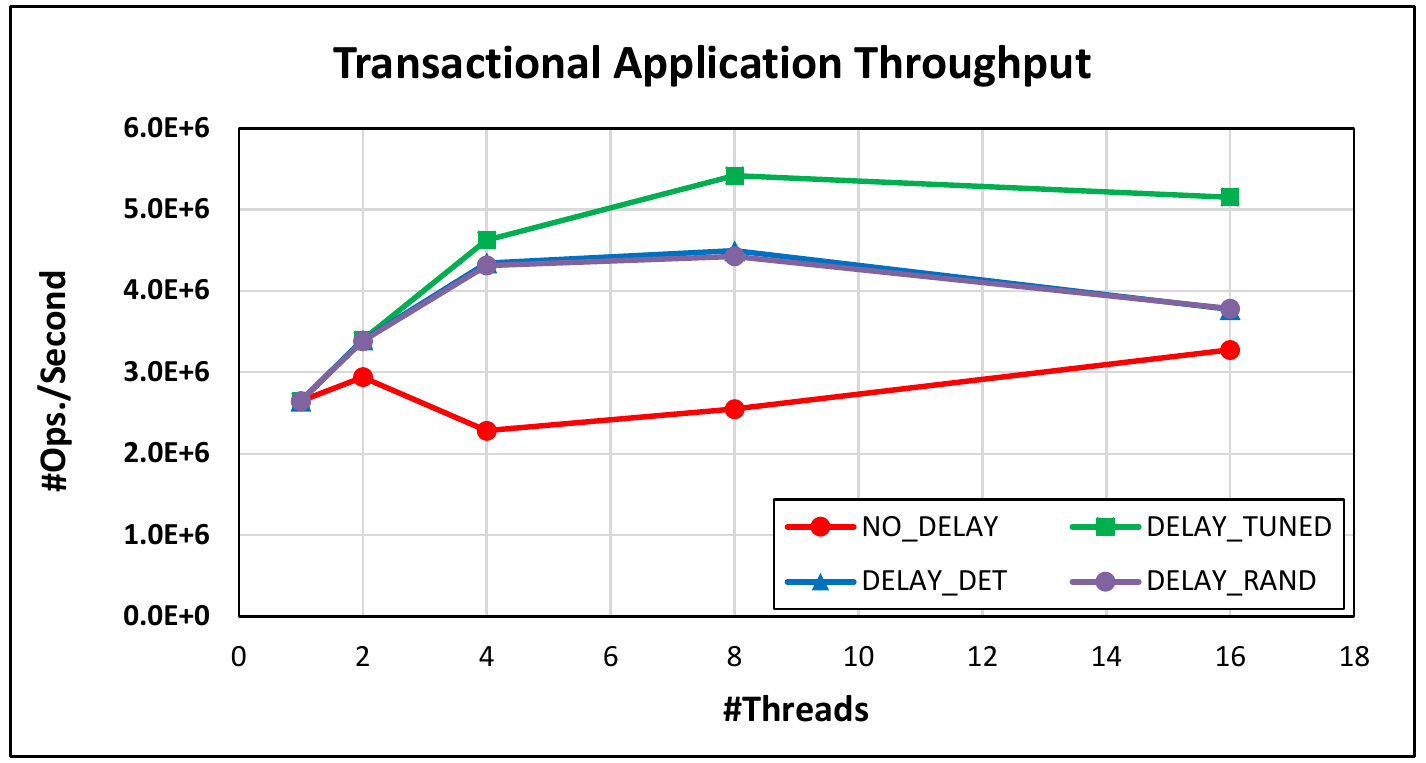}
\end{subfigure}
\begin{subfigure}[b]{0.45\textwidth}
\includegraphics[width=0.9\columnwidth]{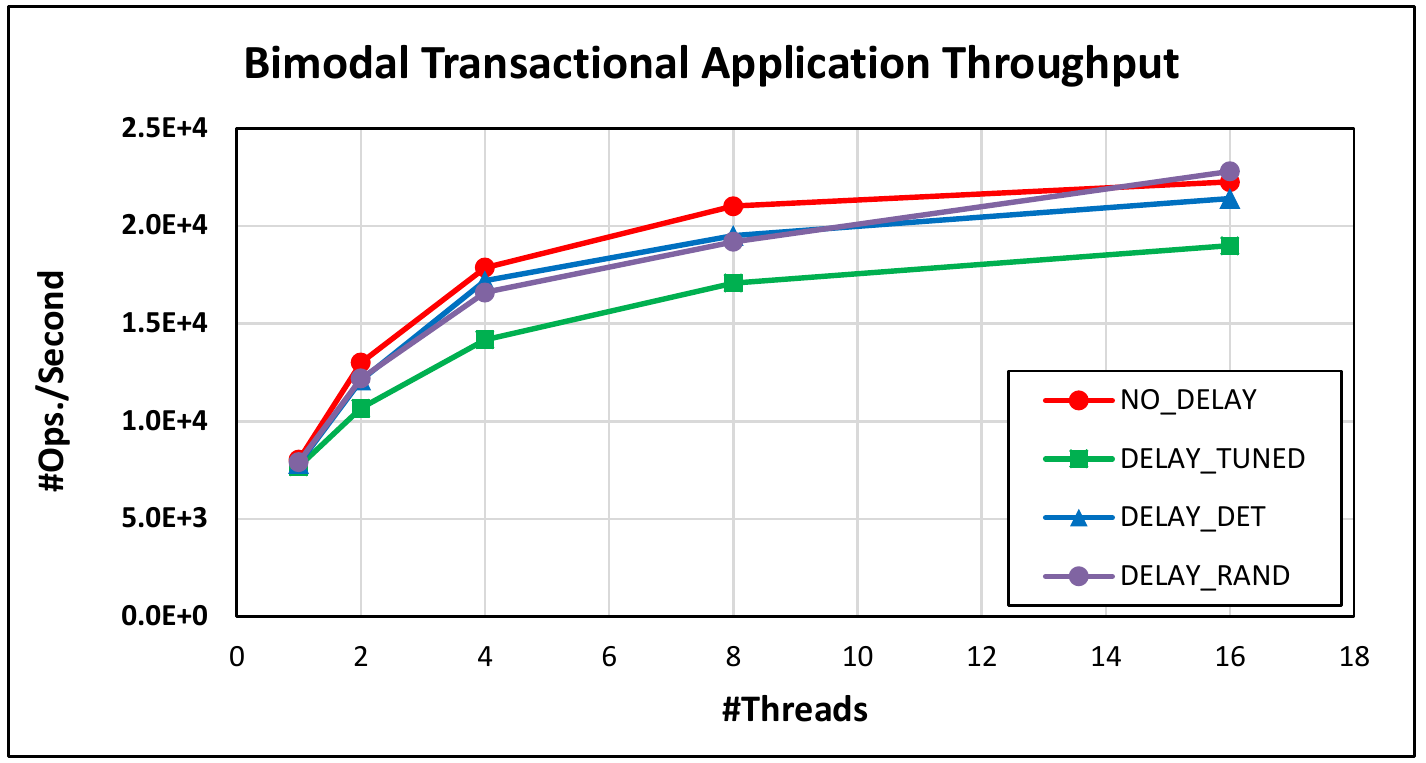}
\end{subfigure}
\caption{Throughput results for lock-free stack, queue, and TL2 benchmark.} 
\label{fig:leasegraphs}
\end{figure*}

Note that $\mu$ stands for the mean of the length distribution and that $B$ denotes the fixed cost that we pay for an abort. Notice that, in Figure \ref{fig:1}, $DET$ performs quite well. The intuitive reason is these distributions are not adversarial. 
Therefore, since we choose a large cost $B$ compared to $\mu$, we (almost) never abort with $DET$. The second observation is that $RRW(\mu)$ and $RRA(\mu)$ perform significantly better than $RRW$ and $RRA$. The reason is that the inequality $\frac{\mu}{B} < 2(\ln 4 -1)$ / $\frac{\mu}{B} \leq 2 \frac{e-2}{e-1}$ holds. The last observation that the cost of $RRW$ and $RRA$ is (almost) exactly $2$, respectively $\frac{e}{e-1}$ times the optimal cost, as predicted by the analysis.

Figure \ref{fig:2} represents the case where the fixed cost $B$ is smaller than $\mu$. This has the following effects on cost: first, $DET$ performs notably worse than before. 
The reason is that $DET$ decides to abort based on the cost, and falls short of this value with higher frequency. 
Also notice that $RRW(\mu)$ and $RRA(\mu)$ and $RRW$ and $RRA$ perform similarly, since the threshold inequality does not hold most of the time. 
Further, we notice that the strategies for requestor aborts, i.e. $RRA(\mu)$ and $RRA$, outperform their respective counterparts from strategy requestor wins.
Figure~\ref{fig:worst_case} illustrates cost under the worst-case distribution for $DET$. 
%

\subsection{Hardware Simulation}

\paragraph{Methodology}
We use Graphite~\cite{graphite}, a tiled multi-core system simulator, to experiment with real applications.
We extend Graphite's directory-based MSI cache coherence protocol for private-L1 shared-L2 cache hierarchy to implement a simple functional equivalent of hardware transactional memory with a requestor-wins policy and lazy validation. 
In particular, the L1 cache controller logic (at the cores) is modified for this purpose, while the directory logic did not have to be modified in any way. 
We test various conflict resolution techniques on top of this setup. 
In particular, we implement the randomized and deterministic delay-setting strategies, as well as a hand-tuned version, which decides on the amount of delay based on knowledge of the dataset and implementation.  
We experiment with two contended data structures implemented using HTM in this setting: a stack and a queue, as well as a simple transactional application. The stack and the queue use lock-free designs as ``slow path" backups. 
The stack and the queue simply alternate inserts and deletes. The transactional application executes transactions which need to jointly acquire and modify two out of a set of $64$ objects in order to commit. 

\paragraph{Results} Please see Figure~\ref{fig:leasegraphs} for the results. The contended data structures are meant to illustrate a setting where the transaction lengths are short and stable. 
In this case, the hand-tuned algorithm does predictably well, since we are able to identify the exact average fast-path length of transactions. Its performance is closely followed by the two online algorithms, which improve significantly on the version which does not implement delays. We observe similar results for the transactional application with uniform transaction lengths. When the transactional application alternates between short and long transactions (the bimodal experiment, in which transactions alternate between short and very long transactions), we notice that  the hand-tuned implementation loses performance, as the transaction length is less predictable. At the same time, the version with no delays performs well, since it tends to abort long transactions, which favors short ones. 
Of note, the randomized algorithm outperforms all other strategies at high contention and high variance, as expected.

\section{Conclusion}

This paper considered the problem of resolving conflicts between transactions using online decision techniques, and presented optimal algorithms for various scenarios arising in real implementations. 
Our results outline a difference between the performance of the requestor aborts and requestor wins strategies in different contention scenarios, and suggest that adding delays can lead to 
improvements in terms of practical performance of transactional systems. In particular, the online algorithms we develop appear to be competitive with offline hand-tuned algorithms. 
We note that the simplicity of the optimal requestor wins strategy--which just chooses a delay uniformly at random within some interval--may lend itself to simple implementation in real systems. 

In terms of future work, we plan to further investigate the practicality of our designs through a more precise HTM implementation, and on a wider series of benchmarks. 
Due to the complexity of implementing a realistic version of RTM~\cite{intel_rtm} accurately on a multiprocessor simulator, this task is beyond the scope of the current work. 
Second, we aim to further refine our conflict model, and to examine whether it is possible to provide guarantees on the competitive ratio under more general conflict models.

\end{document}